\documentclass[compress]{article}
\usepackage{amssymb,cite, amsmath, amsthm}
\usepackage[dvips]{graphicx}
\usepackage{enumerate}
\newtheorem{theorem}{Theorem}[section]

\newtheorem{lemma}[theorem]{Lemma}

\newtheorem{remark}[theorem]{Remark}

\begin{document}

{\large Exact solution of an inverse spherical two phase Stefan problem}
\bigskip

Merey M. Sarsengeldin$^{1,2}$, Abdullah S. Erdogan$^{1}$, Targyn A. Nauryz$^{2}$,  and
Hassan Nouri$^{3}$
\bigskip

$^{1}$ \textit{Sigma LABS, ISE, Kazakh-British Technical University, Almaty,
Kazakhstan }

$^{2}$\textit{\ Institute of Mathematics and Mathematical Modeling, National
Academy of Sciences of Republic of Kazakhstan, Kazakhstan}

$^{3}$ Department of Engineering Design and Mathematics, University of the
West of England, Bristol, UK

\textbf{Abstract} In this paper, we represent the exact solution of a two
phase inverse spherical Stefan problem, where along with unknown temperature
functions heat flux function has to be determined. Suggested solution is
obtained from new form of integral error function and its properties which
are represented in the form of series whose coefficients have to be
determined. The model problem can be used for mathematical modeling and
investigation of arc phenomena in electrical contacts.



\section{Introduction}

Partial differential equations play an important role for the development of
models in heat conduction and investigated in various aspects (see for
example \cite{a,b,c}\ and the references therein). To realize the physical
changes, some models need to be expressed as free or moving boundary
problems. The theory of free boundaries has seen great progress in the last
half century.  For the general literature up to 2015, we refer to \cite{d,e}%
. Present study is devoted to theoretical investigation and mathematical
modeling of arc phenomena in electrical contacts and appears as a
continuation of recent studies where mathematical modeling of short arcing
is considered \cite{1,2}. Arcing processes are very rapid and include phase
transformations like transition therefore it seems reasonable to use Stefan
type problems for mathematical modeling of this phenomena.

Worth to say that exact solution of the problem allows to elucidate and
enhance understanding of arcing processes and contribute to the development
of the arc theory. In the following model heat flux depends on time
variable, however, it is well known that besides time variable heat flux
depends on numerous factors like electron bombardment and diverse electric
contact effects like tunnel, Joule, Thomson, Peltier effects etc. Thus, this
model does not claim to be universal for all electric contact phenomena
occurring during opening or switching electric contacts.

A long list of studies \cite{5,6,7,8,9} and literature therein are devoted to
Stefan type problems and their solutions. In this study, we consider
spherical model which agrees with Holm's so called ideal sphere usually
applied for electric contacts with small contact surface (radius of $%
b<10^{-4}~m$) and low electric current \cite{4}. As it was shown in previous
studies this model nicely fits physical conditions and agrees with
experimental data \cite{3} where we considered AgCdO and Ni contacts \cite{2}.

\subsection{Problem Statement}

\textbf{\ }In a spherical model, the contact spot is given by the spherical
surface of radius $b$. The heat flux $P(t)$ entering this surface melts the
electric contact material (liquid zone $b<r<\alpha (t)$) and then passes
further through the solid zone $\alpha (t)<r<\infty $ (for the illustration
of the model, see \cite[Figure 1]{2}).

The heat equations for each zone are
\begin{equation}
\frac{\partial \theta _{1}}{\partial t}=a_{1}^{2}\left( \frac{\partial
^{2}\theta _{1}}{\partial r^{2}}+\frac{2}{r}\frac{\partial \theta _{1}}{%
\partial r}\right) ,\quad b<r<\alpha (t),\;t>0,  \label{1}
\end{equation}
\begin{equation}
\frac{\partial \theta _{2}}{\partial t}=a_{2}^{2}\left( \frac{\partial
^{2}\theta _{2}}{\partial r^{2}}+\frac{2}{r}\frac{\partial \theta _{2}}{%
\partial r}\right) ,\quad \alpha (t)<r<\infty ,\;t>0  \label{2}
\end{equation}
with initial condition
\begin{equation}
\theta _{1}(b,0)=T_{m},  \label{3}
\end{equation}
\begin{equation}
\theta _{2}(r,0)=f(r),  \label{4}
\end{equation}
\begin{equation}
f(b)=T_{m},\;\alpha (0)=b,\;f(\infty )=0  \label{5}
\end{equation}
subjected to boundary condition at $r=b$
\begin{eqnarray}
-\lambda _{1}\frac{\partial \theta _{1}(b,t)}{\partial r} &=&P(t)  \label{6}
\end{eqnarray}%
and to free boundary
\begin{equation}
\theta _{1}(\alpha (t),t)=T_{m},  \label{7}
\end{equation}
\begin{equation}
\theta _{2}(\alpha (t),t)=T_{m}  \label{8}
\end{equation}
the Stefan's condition
\begin{equation}
-\lambda _{1}\frac{\partial \theta _{1}(\alpha (t),t)}{\partial r}=-\lambda
_{2}\frac{\partial \theta _{2}(\alpha (t),t)}{\partial r}+L\gamma \frac{%
d\alpha }{dt},  \label{9}
\end{equation}
as well as the condition at infinity
\begin{equation}
\theta _{2}(\infty ,t)=0.  \label{10}
\end{equation}%
Here $\theta _{1}$ and $\theta _{2}$ are unknown heat functions, $P(t)$ is
an unknown heat flux coming from electric arc, $T_{m}$ is a melting
temperature of electric contact material, $f(r)$ is a given function and $%
a_{1}$, $a_{2}$, $\lambda _{1}$, $\lambda _{2}$ and $L\gamma $ are given
constants. In the equation, power balance is described by Stefan's condition
$\left( \ref{9}\right) $.

\section{Problem Solution}

Supposing that the initial and free boundary are analytic functions and can
be expanded in Taylor series as

\begin{equation}
f(r)=\sum_{n=0}^{\infty }\frac{f^{(n)}(b)}{n!}(r-b)^{n},\quad \alpha
(t)=b+\sum_{n=1}^{\infty }\alpha _{n}t^{n/2},  \label{fr}
\end{equation}%
we represent the solution of $\left( \ref{1}\right) -\left( \ref{10}\right) $
in the following form

\begin{equation*}
\theta _{1}(r,t)=\frac{1}{r}\sum_{n=0}^{\infty }(2a_{1}\sqrt{t})^{n}\left[
A_{n}i^{n}erfc\frac{r-b}{2a_{1}\sqrt{t}}+B_{n}i^{n}erfc\frac{b-r}{2a_{1}%
\sqrt{t}}\right] ,
\end{equation*}

\begin{equation*}
\theta _{2}(r,t)=\frac{1}{r}\sum_{n=0}^{\infty }(2a_{2}\sqrt{t})^{n}\left[
C_{n}i^{n}erfc\frac{r-b}{2a_{2}\sqrt{t}}+D_{n}i^{n}erfc\frac{b-r}{2a_{2}%
\sqrt{t}}\right] ,
\end{equation*}
where coefficients $A_{n} ,B_{n} ,\, C_{n} $and $D_{n} $ have to be found.

\begin{lemma} \label{lemma} The integral error function holds the following properties:

\begin{enumerate}[(a)]
\item $\mathop{\lim }\limits_{x\rightarrow \infty }\frac{i^{n}erfc(-x)}{x^{n}%
}=\frac{2}{n!},$

\item ${\mathop{\lim }\limits_{t\rightarrow 0}(2a_{1}\sqrt{t}%
)^{n}C_{n}i^{n}erfc\frac{r-b}{2a_{1}\sqrt{t}}=0,}$

\item ${\mathop{\lim }\limits_{t\rightarrow 0}(2a_{1}\sqrt{t}%
)^{n}D_{n}i^{n}erfc\frac{b-r}{2a_{1}\sqrt{t}}=\frac{2}{n!}D_{n}(r-b)^{n}.}$
\end{enumerate}
\end{lemma}
The proof of the lemma can be given by the L'Hopital's rule and properties
of $i^{n}erfcx$ function.

\begin{theorem} Let $f$  be $n$ times differentiable analytic function.
Then

\begin{equation*}
\lim\limits_{t\rightarrow 0}\theta _{2}(r,t)=\frac{1}{r}\sum_{n=0}^{\infty }%
\frac{2}{n!}D_{n}(r-b)^{n}.
\end{equation*}
\end{theorem}
\begin{proof} Using Lemma \ref{lemma}, it is easy to see that%
\begin{equation*}
\lim\limits_{t\rightarrow 0}\theta _{2}(r,t)=\lim\limits_{t\rightarrow 0}%
\frac{1}{r}\sum_{n=0}^{\infty }(2a_{2}\sqrt{t})^{n}{D_{n}i^{n}erfc\frac{b-r}{%
2a_{2}\sqrt{t}}}
\end{equation*}%
\begin{equation*}
=\frac{1}{r}\lim\limits_{t\rightarrow 0}\sum_{n=0}^{\infty }D_{n}\left(
r-b\right) ^{n}\frac{{i^{n}erfc}\left( -{\frac{r-b}{2a_{2}\sqrt{t}}}\right)
}{\left( {\frac{r-b}{2a_{2}\sqrt{t}}}\right) ^{n}}=\frac{1}{r}%
\sum_{n=0}^{\infty }\frac{2}{n!}D_{n}\left( r-b\right) ^{n}.
\end{equation*}
\end{proof}
\subsection{Calculation of coefficients}

\bigskip By the theorem and equations $\left( \ref{4}\right) $ and $\left( %
\ref{fr}\right) $, we can write
\begin{equation*}
\lim\limits_{t\rightarrow 0}\theta _{2}(r,t)=\frac{1}{r}\sum_{n=0}^{\infty }%
\frac{2}{n!}D_{n}\left( r-b\right) ^{n}=f\left( r\right)
=\sum_{n=0}^{\infty }\frac{f^{(n)}(b)}{n!}(r-b)^{n}.
\end{equation*}%
Thus, we get
\begin{equation*}
D_{n}=\frac{1}{2}rf^{(n)}(b).
\end{equation*}

From $\left( \ref{8}\right) $ when we put $r=\alpha (t)$ then $b$ will be
canceled and there left only series $\alpha (t)=\alpha _{1}\sqrt{t}+\alpha
_{2}t+\alpha _{3}\sqrt{t^{3}}+...=\sum\limits_{n=1}^{\infty }\alpha
_{1}t^{n/2}$ in Hartree function if we take $\sqrt{t}$ out of the brackets

\begin{equation*}
i^{n}erfc\frac{\sqrt{t}(\alpha _{1}+\alpha _{2}\sqrt{t}+\alpha _{3}t+...)}{%
2a_{1}\sqrt{t}}=i^{n}erfc\frac{\alpha (t)}{2a_{1}},
\end{equation*}%
where
\begin{equation*}
\alpha (t)=\alpha _{1}+\alpha _{2}\sqrt{t}+\alpha _{3}t+\cdots
=\sum\limits_{n=1}^{\infty }\alpha _{n}t^{\frac{n-1}{2}}.
\end{equation*}

Let's take $\sqrt{t}=\tau $ and we obtain from equation $\left( \ref{8}%
\right) $

\begin{equation}
\frac{1}{\alpha (\tau )}\sum_{n=0}^{\infty }(2a_{2}\sqrt{t})^{n}\left[
C_{n}i^{n}erfc\delta (\tau )+D_{n}i^{n}erfc(-\delta (\tau ))\right] =T_{m},
\label{14}
\end{equation}

where $\delta (\tau )=\frac{\alpha (\tau )}{2a_{2}}.$

To calculate coefficient $C_{n} $ we apply Leibniz, Faa Di Bruno's formulas
and Bell polynomials. Using Leibniz formula we have

\begin{equation*}
\frac{\partial ^{k}[2^{n/2}\tau ^{n}i^{n}erfc\delta ]}{\partial \tau ^{k}}%
\left\vert
\begin{array}{l}
\\
{\tau =0}%
\end{array}%
\right. =\left\{
\begin{array}{l}
{0,\;for\;k<n} \\
{\frac{2^{n/2}k!}{(k-n)!}[i^{n}erfc\delta ]^{(k-n)},\;for\;k\geq n}%
\end{array}%
\right. .\;
\end{equation*}

Using Faa Di Bruno's formula and Bell polynomials for a derivative of a
composite function we have

\begin{equation*}
\frac{\partial ^{(k-n)}[i^{n}erfc(\pm \delta )]}{\partial \tau ^{k-n}}%
\left\vert
\begin{array}{l}
\\
{\tau =0}%
\end{array}%
\right. =\sum_{m=1}^{k-n}(i^{n}erfc(\pm \delta ))^{(m)}\left\vert
\begin{array}{l}
\\
{\delta =0}%
\end{array}%
\right. B_{k-n,m}(\delta ^{\prime (k-n-m+1)}(\tau ))|_{\tau =0},
\end{equation*}

where
\begin{equation*}
B_{k-n,m}=\sum \frac{(k-n)!}{j_{1}!j_{2}!...j_{k-n-m+1}!}\cdot \beta
_{1}^{j_{1}}\beta _{2}^{j_{2}}\beta _{3}^{j_{3}}\cdots \beta
_{k-n-m+1}^{j_{k-n-m+1}}
\end{equation*}%
and $j_{1},j_{2},\cdots $ satisfy the following equations

\begin{equation*}
\begin{array}{l}
{j_{1}+j_{2}+...+j_{k-n-m+1}=m,} \\
{j_{1}+2j_{2}+...+(k-n-m+1)j_{k-n-m+1}=k-n}%
\end{array}%
\end{equation*}
for $m\geq n$

\begin{equation*}
[i^{n} erfc(\pm \delta )]^{(m)} |_{\alpha =0} =(-1)^{m} i^{n-m} erfc0=(\mp
1)^{m} \frac{\Gamma \left(\frac{n-m+1}{2} \right)}{(n-m!)\sqrt{\pi } }
\end{equation*}

by taking both sides of $\left( \ref{14}\right) $ k-times derivatives at $%
\tau =0$ we have

\begin{equation}
\begin{array}{l}
{\left\{ \sum\limits_{m=1}^{k-p-n}(-1)^{m}\frac{1}{\alpha ^{m+1}}\sum \frac{%
n!\alpha _{1}^{j_{1}}\alpha _{2}^{j_{2}}...\alpha _{n-m+1}^{j_{n-m+1}}}{%
j_{1}!j_{2}!...j_{n-m+1}!}\,\right\} \cdot \sum\limits_{p=0}^{k}\left(
\begin{array}{l}
{k} \\
{p}%
\end{array}%
\right) \left\{ \sum\limits_{n=0}^{\infty }\frac{2^{n/2}(k-p)!}{(k-p-n)!}%
\right. } \\
{\left( (-1)^{n}C_{n}\sum_{m=1}^{k-p-n}i^{n-m}erfc\delta _{1}\sum \frac{%
(k-p)!\delta _{1}^{j_{1}}\delta _{2}^{j_{2}}...\delta
_{k-p-n-m+1}^{j_{k-p-n-m+1}}}{j_{1}!j_{2}!...j_{k-p-n-m+1}!}\right. } \\
{\left. \left. {+}D_{n}\sum_{m=1}^{k-p-n}i^{n-m}erfc(-\delta _{1})\sum \frac{%
(k-p)!\delta _{1}^{j_{1}}\delta _{2}^{j_{2}}...\delta
_{k-p-n-m+1}^{j_{k-p-n-m+1}}}{j_{1}!j_{2}!...j_{k-p-n-m+1}!}\right) \right\}
=0.}%
\end{array}
\label{15}
\end{equation}

From expression $\left( \ref{15}\right) $ we express coefficients $C_{n}$.
From $\left( \ref{7}\right) $ condition we have

\begin{equation*}
\frac{1}{\alpha (\tau )}\sum_{n=0}^{\infty }(2a_{1}\sqrt{t})^{n}\left[
A_{n}i^{n}erfc\xi (\tau )+B_{n}i^{n}erfc(-\xi (\tau ))\right] =T_{m}
\end{equation*}
where\qquad
\begin{equation}
\xi (\tau )=\frac{\alpha (\tau )}{2a_{1}}.  \label{16}
\end{equation}

As previously by taking both sides of $\left( \ref{16}\right) $ k-times
derivatives $\tau =0$ by using Leibniz, Faa Di Bruno's formulas and Bell
polynomials we have

\begin{equation}
\begin{array}{l}
{\left\{ \sum_{m=1}^{k-p-n}(-1)^{m}\frac{1}{\alpha ^{m+1}}\sum \frac{%
n!\alpha _{1}^{j_{1}}\alpha _{2}^{j_{2}}...\alpha _{n-m+1}^{j_{n-m+1}}}{%
j_{1}!j_{2}!...j_{n-m+1}!}\,\right\} \cdot \sum_{p=0}^{k}\left(
\begin{array}{l}
{k} \\
{p}%
\end{array}%
\right) \left\{ \sum_{n=0}^{\infty }\frac{2^{n/2}(k-p)!}{(k-p-n)!}\right. }
\\
{\left( (-1)^{n}A_{n}\sum_{m=1}^{k-p-n}i^{n-m}erfc\xi _{1}\sum \frac{%
(k-p)!\xi _{1}^{j_{1}}\xi _{2}^{j_{2}}...\xi _{k-p-n-m+1}^{j_{k-p-n-m+1}}}{%
j_{1}!j_{2}!...j_{k-p-n-m+1}!}\right. } \\
{\left. \left. {+}B_{n}\sum_{m=1}^{k-p-n}i^{n-m}erfc(-\xi _{1})\sum \frac{%
(k-p)!\xi _{1}^{j_{1}}\xi _{2}^{j_{2}}...\xi _{k-p-n-m+1}^{j_{k-p-n-m+1}}}{%
j_{1}!j_{2}!...j_{k-p-n-m+1}!}\right) \right\} =0.}%
\end{array}
\label{17}
\end{equation}

We can express from this expression $A_{n}$ coefficient. In Stefan's
condition $\left( \ref{9}\right) $ we take first derivative of
\begin{equation*}
\frac{d\alpha }{dt}=\alpha _{1}\frac{1}{2\sqrt{t}}+\alpha _{2}+\frac{3}{2}%
\alpha _{3}\sqrt{t}+2\alpha _{4}t+...=\sum_{n=1}^{\infty }\frac{n}{2}\alpha
_{n}t^{\frac{n-2}{2}}
\end{equation*}%
and we take $\sqrt{t}=\tau $

\begin{equation}
\begin{array}{l}
{\lambda _{1}\frac{1}{\alpha ^{2}(\tau )}\sum_{n=0}^{\infty }(2a_{1}\tau
)^{n}\left[ A_{n}i^{n}erfc\xi (\tau )+B_{n}i^{n}erfc(-\xi (\tau ))\right] }
\\
{-\lambda _{1}\frac{1}{\alpha (\tau )}\sum_{n=0}^{\infty }(2a_{1}\tau )^{n-1}%
\left[ -A_{n}i^{n-1}erfc\xi (\tau )+B_{n}i^{n-1}erfc(-\xi (\tau ))\right] }
\\
{=\lambda _{2}\frac{1}{\alpha ^{2}(\tau )}\sum_{n=0}^{\infty }(2a_{2}\tau
)^{n}\left[ C_{n}i^{n}erfc\delta (\tau )+D_{n}i^{n}erfc(-\delta (\tau ))%
\right] } \\
{-\lambda _{2}\frac{1}{\alpha (\tau )}\sum_{n=0}^{\infty }(2a_{1}\tau )^{n-1}%
\left[ -C_{n}i^{n-1}erfc\delta (\tau )+D_{n}i^{n-1}erfc(-\delta (\tau ))%
\right] +\alpha (\tau )L\gamma }%
\end{array}
\label{18}
\end{equation}

where $\alpha (\tau )=\sum_{n=1}^{\infty }\frac{n}{2}\alpha _{n}\tau ^{n-2}.$

If multiply both sides of $\left( \ref{18}\right) $ by $\alpha (\tau )$ and
using conditions $\left( \ref{7}\right) $,$\left( \ref{8}\right) $ we have
the following expression

\begin{equation}
\begin{array}{l}
{\lambda _{1}\frac{T_{m}}{\alpha (\tau )}-\lambda _{1}\sum_{n=0}^{\infty
}(2a_{1}\tau )^{n-1}\left[ -A_{n}i^{n-1}erfc\xi (\tau
)+B_{n}i^{n-1}erfc(-\xi (\tau ))\right] } \\
{=\lambda _{2}\frac{T_{m}}{\alpha (\tau )}-\lambda _{2}\sum_{n=0}^{\infty
}(2a_{2}\tau )^{n-1}\left[ -C_{n}i^{n-1}erfc\delta (\tau
)+D_{n}i^{n-1}erfc(-\delta (\tau ))\right] +\beta (\tau )L\gamma ,}%
\end{array}
\label{19}
\end{equation}

where
\begin{equation*}
\beta (\tau )=\frac{b}{2}\sum_{n=1}^{\infty }n\alpha _{n}\tau ^{n-2}+\frac{1%
}{2}\sum_{n=1}^{\infty }n\alpha _{n}\tau ^{2(n-1)}.
\end{equation*}

Previously by taking both sides of $\left( \ref{19}\right) $ k-times
derivatives at $\tau =0$ by using Leibniz, Faa Di Bruno's formulas and Bell
polynomials we obtain

\begin{equation}
\begin{array}{l}
{\left\{ \lambda _{1}T_{m}\sum_{m=1}^{k-p-n}(-1)^{m}\frac{1}{\alpha ^{m+1}}%
\sum \frac{n!\alpha _{1}^{j_{1}}\alpha _{2}^{j_{2}}...\alpha
_{n-m+1}^{j_{n-m+1}}}{j_{1}!j_{2}!...j_{n-m+1}!}\,\right\} \cdot
\sum_{p=0}^{k}\left(
\begin{array}{l}
{k} \\
{p}%
\end{array}%
\right) \left\{ \sum_{n=0}^{\infty }\frac{2^{n/2}(k-p)!}{(k-p-n-1)!}\right. }
\\
{\left( (-1)^{n+1}A_{n}\sum_{m=1}^{k-p-n-1}i^{n-m-1}erfc\xi _{1}\sum \frac{%
(k-p)!\xi _{1}^{j_{1}}\xi _{2}^{j_{2}}...\xi _{k-p-n-m}^{j_{k-p-n-m}}}{%
j_{1}!j_{2}!...j_{k-p-n-m}!}\right. } \\
{\left. \left. {+}B_{n}\sum_{m=1}^{k-p-n-1}i^{n-m-1}erfc(-\xi _{1})\sum
\frac{(k-p)!\xi _{1}^{j_{1}}\xi _{2}^{j_{2}}...\xi _{k-p-n-m}^{j_{k-p-n-m}}}{%
j_{1}!j_{2}!...j_{k-p-n-m}!}\right) \right\} } \\
{=\left\{ \lambda _{2}T_{m}\sum_{m=1}^{k-p-n}(-1)^{m}\frac{1}{\alpha ^{m+1}}%
\sum \frac{n!\alpha _{1}^{j_{1}}\alpha _{2}^{j_{2}}...\alpha
_{n-m+1}^{j_{n-m+1}}}{j_{1}!j_{2}!...j_{n-m+1}!}\,\right\} \cdot
\sum_{p=0}^{k}\left(
\begin{array}{l}
{k} \\
{p}%
\end{array}%
\right) \left\{ \sum_{n=0}^{\infty }\frac{2^{n/2}(k-p)!}{(k-p-n-1)!}\right. }
\\
{\left( (-1)^{n+1}C_{n}\sum_{m=1}^{k-p-n-1}i^{n-m-1}erfc\delta _{1}\sum
\frac{(k-p)!\delta _{1}^{j_{1}}\delta _{2}^{j_{2}}...\delta
_{k-p-n-m}^{j_{k-p-n-m}}}{j_{1}!j_{2}!...j_{k-p-n-m}!}\right. } \\
{+\left. \left. D_{n}\sum_{m=1}^{k-p-n-1}i^{n-m-1}erfc(-\delta _{1})\sum
\frac{(k-p)!\delta _{1}^{j_{1}}\delta _{2}^{j_{2}}...\delta
_{k-p-n-m}^{j_{k-p-n-m}}}{j_{1}!j_{2}!...j_{k-p-n-m}!}\right) \right\} +%
\frac{n!}{2}\left( b\alpha _{k}+\alpha _{k}^{2}\right) L\gamma .}%
\end{array}
\label{20}
\end{equation}

From this recurrent formula we can express $B_{n}$. By using all these
expression on condition $\left( \ref{6}\right) $ we express coefficient of
heat flux. From $\left( \ref{6}\right) $, we get

\begin{equation}
\begin{array}{l}
{\frac{\lambda _{1}}{b^{2}}\sum_{n=0}^{\infty }(2a_{1}\sqrt{t})^{n}\left[
A_{n}i^{n}erfc0+B_{n}i^{n}erfc0\right] } \\
{\quad \quad +\frac{\lambda _{1}}{b}\sum_{n=0}^{\infty }(2a_{1}\sqrt{t}%
)^{n-1}\left[ A_{n}i^{n-1}erfc0-B_{n}i^{n-1}erfc0\right] =\sum_{n=0}^{\infty
}P_{n}t^{n/2},}%
\end{array}
\label{21}
\end{equation}

\begin{equation*}
\left\{
\begin{array}{l}
{P_{0}=\frac{\lambda _{1}}{b}i^{0}erfc0[A_{1}-B_{1}+\frac{1}{b}A_{0}+\frac{1%
}{b}B_{0}],} \\
{P_{1}=\frac{\lambda _{1}}{b}(2a_{1})i^{1}erfc0[A_{2}-B_{2}+\frac{1}{b}A_{1}+%
\frac{1}{b}B_{1}],} \\
{P_{2}=\frac{\lambda _{1}}{b}(2a_{1})^{2}i^{2}erfc0[A_{3}-B_{3}+\frac{1}{b}%
A_{2}+\frac{1}{b}B_{2}],} \\
\vdots \\
{P_{n}=\frac{\lambda _{1}}{b}(2a_{1})^{n}i^{n}erfc0[A_{n+1}-B_{n+1}+\frac{1}{%
b}A_{n}+\frac{1}{b}B_{n}].}%
\end{array}%
\right.
\end{equation*}

\begin{remark} For the convergence of temperature
functions $\Theta _{1},\Theta _{2}$, it is possible
to follow the idea proposed in \cite{2}.
\end{remark}
\section{Approximate solution of a test problem}

To show the effectiveness of the method we proposed the following problem.
In this section we show that it is possible to reach $3.5\%$ error using
only three points $t_{0}=0$, $t_{1}=0.5$ and $t_{2}=1$ by collocation method
which seems to be very practical for engineers.

It is well known that by the help of substitution $\Theta _{i}=\frac{U_{i}}{r%
}$ it is not difficult to show that spherical heat equation $\frac{\partial
\theta _{i}}{\partial t}=a_{i}^{2}\left( \frac{\partial ^{2}\theta _{i}}{%
\partial r^{2}}+\frac{2}{r}\frac{\partial \theta _{i}}{\partial r}\right) $
can be reduced to the well-known heat equation $\frac{\partial u_{i}}{%
\partial t}=a_{i}\frac{\partial ^{2}u_{i}}{\partial x^{2}}$ where $i=1,2$.
Thus problem $\left( \ref{1}\right) $-$\left( \ref{10}\right) $ can be
reduced to the following problem with slightly modified and simplified
boundary conditions without the loss of generality. Solution is found both
exactly and approximately. Let us consider

\begin{equation}  \label{22}
\begin{array}{l}
{\frac{\partial u_{1} }{\partial t} =a_{1} \frac{\partial ^{2} u_{1} }{%
\partial x^{2} } ,\; 0<x<\alpha (t),\; t>0,} \\
{\frac{\partial u_{2} }{\partial t} =a_{2} \frac{\partial ^{2} u_{2} }{%
\partial x^{2} } ,\; \alpha (t)<x<\infty ,\; t>0,}%
\end{array}%
\end{equation}

\begin{equation}
u_{1}(0,0)=0,\;\,u_{2}(x,0)=f(x),\text{ }f(0)=u_{m},  \label{23}
\end{equation}
\begin{equation}
u_{1}(\alpha (t),t)=u_{2}(\alpha (t),t)=u_{m},  \label{24}
\end{equation}
\begin{equation}
u_{1}(0,t)=P(t),  \label{25}
\end{equation}
\begin{equation}
-\lambda _{1}\frac{\partial u_{1}}{\partial x}=-\lambda _{2}\frac{\partial
u_{2}}{\partial x}+L\gamma \frac{d\alpha }{dt}.  \label{26}
\end{equation}

Here $u_{m}$ is a constant and $\alpha (t)=\alpha \sqrt{t}$ is chosen as a
moving boundary.

We represent solution of $\left( \ref{22}\right) $-$\left( \ref{26}\right) $
in the following form

\begin{equation*}
\begin{array}{l}
{u_{1}(x,t)=\sum_{n=0}^{\infty }(2a_{1}\sqrt{t})^{n}\left[ A_{n}i^{n}erfc%
\frac{x}{2a_{1}\sqrt{t}}+B_{n}i^{n}erfc\frac{-x}{2a_{1}\sqrt{t}}\right] ,}
\\
{u_{2}(x,t)=\sum_{n=0}^{\infty }(2a_{2}\sqrt{t})^{n}\left[ C_{n}i^{n}erfc%
\frac{x}{2a_{2}\sqrt{t}}+D_{n}i^{n}erfc\frac{-x}{2a_{2}\sqrt{t}}\right] ,}%
\end{array}%
\end{equation*}
and
\begin{equation*}
P(t)=\sum\limits_{n=0}^{\infty }P_{n}(\sqrt{t})^{n},
\end{equation*}
where $P_{n},$ $A_{n},$ $B_{n},$ $C_{n},$ and $D_{n}$ have to be determined.
Following the same analogy given in the previous section we find
\begin{equation}
D_{n}=\frac{1}{2}f^{(n)}(0),\;\;n=0,1,2,\cdots ,  \label{27}
\end{equation}
\begin{equation}
C_{0}=\frac{1}{2}u_{m},  \label{28}
\end{equation}
\begin{equation}
C_{n}=-\frac{1}{2}f^{(n)}(0)\frac{i^{n}erfc\frac{-\alpha }{2a_{2}}}{i^{n}erfc%
\frac{\alpha }{2a_{2}}},  \label{29}
\end{equation}
\begin{equation}
A_{0}=\frac{u_{m}-B_{0}erfc\frac{-\alpha }{2a_{1}}}{erfc\frac{\alpha }{2a_{1}%
}}.  \label{30}
\end{equation}
From $\left( \ref{26}\right) $ we get
\begin{equation}
\begin{array}{l}
{\lambda _{1}(2a_{1})^{-1}\frac{2}{\sqrt{\pi }}\left( -A_{0}\exp \left(
\frac{\alpha }{2a_{1}}\right) ^{2}+B_{0}\exp \left( -\left( \frac{\alpha }{%
2a_{1}}\right) ^{2}\right) \right) } \\
{=\lambda _{2}(2a_{2})^{-1}\frac{2}{\sqrt{\pi }}\left( -C_{0}\exp \left(
\frac{\alpha }{2a_{2}}\right) ^{2}+D_{0}\exp \left( -\left( \frac{\alpha }{%
2a_{2}}\right) ^{2}\right) \right) +\frac{L\gamma \alpha }{2}}%
\end{array}
\label{31}
\end{equation}
and from $\left( \ref{24}\right) $ we get
\begin{equation}
A_{n}=-B_{n}\frac{i^{n}erfc\frac{-\alpha }{2a_{1}}}{i^{n}erfc\frac{\alpha }{%
2a_{1}}}.  \label{32}
\end{equation}
Again from Stefan's condition expressed by equation $\left( \ref{26}\right) $%
, we get
\begin{eqnarray}
&&-\lambda _{1}(2a_{1})^{n-1}\left[ -A_{n}i^{n-1}erfc\frac{\alpha }{2a_{1}}%
+B_{n}i^{n-1}erfc\frac{-\alpha }{2a_{1}}\right]  \notag \\
&=&-\lambda _{2}(2a_{2})^{n-1}\left[ -C_{n}i^{n-1}erfc\frac{\alpha }{2a_{2}}%
+D_{n}i^{n-1}erfc\frac{-\alpha }{2a_{2}}\right] ,  \label{33}
\end{eqnarray}
\begin{equation}
P_{n-1}=\lambda _{1}(2a_{1})^{n-1}(A_{n}-B_{n})i^{n-1}erfc0,n=1,2,3,\cdots .
\label{34}
\end{equation}
Thus coefficients $A_{0},$ $B_{0}$ can be determined from equations $\left( %
\ref{30}\right) $, $\left( \ref{31}\right) $, coefficients $A_{n},$ $B_{n}$
from $\left( \ref{32}\right) $ and $\left( \ref{33}\right) $,$%
C_{0},C_{n},D_{n}$ from $\left( \ref{28}\right) $, $\left( \ref{29}\right) $
and $\left( \ref{27}\right) $ respectively and $P_{n}$ from $\left( \ref{34}%
\right) $.

\subsection{Test problem}

Let $f(x)=u_{m} +x,\; u_{m} =0,\lambda _{1} =\lambda _{2} =1,a_{1} =a_{2}
=1,L=\alpha =\gamma =1$.

We use Mathcad 15 for calculations and get following approximate values for $%
A_0=0.579$, $A_1=-3.004$, $A_2=-2.333\times 10^{-15}$, $B_0=-0.183$, $%
B_1=0.5 $, $B_2=0$ whereas exact values are $A_0=0.604$, $A_1=-3.004$, $%
A_2=0 $, $B_0=-0.191$, $B_1=0.5$, and $B_2=0$.

\begin{figure}[h!]
\centering
\includegraphics[width=6cm]{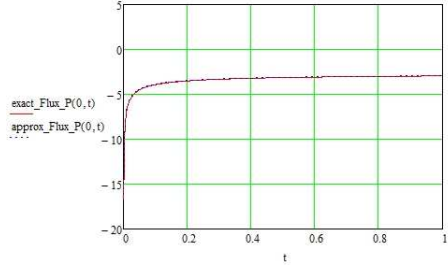}
\caption{Exact and approximate values of flux function}
\end{figure}

\begin{figure}[h!]
\centering
\includegraphics[width=6cm]{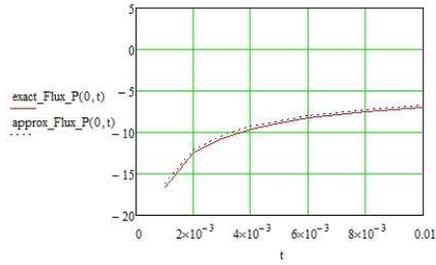}
\caption{Exact and approximate values of flux function for small t values}
\end{figure}

\begin{figure}[h!]
\centering
\includegraphics[width=6cm]{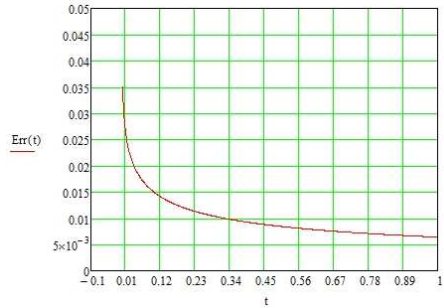}
\caption{Relative error}
\end{figure}

In Figures 1 and 2, the graphs of both reconstructed exact (\textbf{%
exact\_Flux\_P(0,t)}) and approximate (\textbf{approx\_Flux\_P(0,t)}) flux
functions are shown.

In Figure 3 we illustrate the graph of relative error function calculated by
following formula which reaches maximum value $3.5\%$ at point $x=0,~0\leq
t\leq 1$%
\begin{equation*}
Error_{rel}=\frac{\left\vert \text{exact flux-approximate flux}\right\vert }{%
\text{exact flux}}.
\end{equation*}

\section{Conclusion}

A mathematical model describing heat propagation in electric contacts is
constructed on the base of two phase spherical inverse Stefan problem. The
heat source P(t)is determined from equation $\left( \ref{21}\right) $, which
is due to our assumption could be arcing, bridging, Joule heating etc.
Temperature functions $\Theta _{1},\Theta _{2}$ which are given in the form
of series are determined whose coefficients $\;A_{n},\;B_{n},C_{n}\;$ and $%
D_{n}$ are also determined from equations $\left( \ref{15}\right) $, $\left( %
\ref{17}\right) $, $\left( \ref{18}\right) $ and $\left( \ref{20}\right) $.
In the test problem we used maximum principle for error estimate, the
deviation is 3.5\% for three points. For better precision more points has to
be taken and better computer characteristics are required.

\section*{Acknowledgements}
 This research is supported by the Ministry of Science and Education of the
Republic of Kazakhstan, Grant Number 0115RK00653. The authors would like
to thank Prof. S. N. Kharin (IMMM and KBTU, Kazakhstan) for their valuable comments and suggestions which were
helpful in improving the paper.


\begin{thebibliography}{99}
\bibitem{2} Sarsengeldin, MM, Kharin, SN: Method of the integral error
functions for the solution of the one- and two-phase Stefan problems and its
application. Filomat (accepted).
\end{thebibliography}
\end{document}